\documentclass[
  aps,
  pra,
  reprint,
  letterpaper,
  amsmath,
  amssymb,
  superscriptaddress,
  longbibliography
]{revtex4-2}

\usepackage{mathtools}
\usepackage{amsthm}
\usepackage{microtype}
\usepackage{graphicx}
\usepackage{mathrsfs}
\usepackage[colorlinks=true,allcolors=blue]{hyperref}

\allowdisplaybreaks

\newtheorem{theorem}{Theorem}
\newtheorem{proposition}[theorem]{Proposition}

\newtheorem{corollary}[theorem]{Corollary}
\theoremstyle{remark}

\newcommand{\Q}{\mathcal{Q}}
\newcommand{\Qp}{\mathcal{Q}_{\mathrm{PVM}}}
\newcommand{\Tr}{\operatorname{Tr}}
\newcommand{\id}{\mathbb{I}}
\newcommand{\ket}[1]{\lvert #1\rangle}
\newcommand{\bra}[1]{\langle #1\rvert}
\newcommand{\abs}[1]{\lvert #1\rvert}
\newcommand{\norm}[1]{\lVert #1\rVert}
\newcommand{\Bellop}{\widehat{\mathcal{B}}}
\newcommand{\Gram}{\boldsymbol{\Gamma}}
\newcommand{\Ksqrt}{\mathbb{Q}(\sqrt{2})}

\makeatletter
\newcommand{\equalcontrib}{%
  \frontmatter@footnote{L.Z. and R.C. contributed equally to this work.}%
}
\makeatother

\begin{document}

\title{Analytic Qubit Separation between POVMs and Projective Measurements}

\author{Lin Zhu\equalcontrib}
\affiliation{Thrust of Artificial Intelligence, Information Hub,
The Hong Kong University of Science and Technology (Guangzhou),
Guangdong 511453, China}
\affiliation{Quantum Science Center of Guangdong--Hong Kong--Macao
Greater Bay Area, Shenzhen 518045, China}

\author{Ranyiliu Chen\equalcontrib}
\email{chenranyiliu@quantumsc.cn}
\affiliation{Quantum Science Center of Guangdong--Hong Kong--Macao
Greater Bay Area, Shenzhen 518045, China}

\author{Xin Wang}
\affiliation{Thrust of Artificial Intelligence, Information Hub,
The Hong Kong University of Science and Technology (Guangzhou),
Guangdong 511453, China}

\author{Shenggen Zheng}
\email{zhengshenggen@quantumsc.cn}
\affiliation{Quantum Science Center of Guangdong--Hong Kong--Macao
Greater Bay Area, Shenzhen 518045, China}

\begin{abstract}
Generalized measurements can be implemented projectively after enlarging the Hilbert space, but this dilation changes the available local dimension. We construct a Bell functional with rational coefficients that separates the two measurement models at local dimension two. An explicit three-outcome qubit positive-operator-valued measure with rational matrix entries attains \(2\sqrt2+1/100\). On the other hand, all qubit-projective strategies are bounded by \(2\sqrt2+\sqrt5/250+\sqrt2/32400\), giving a fully analytic certified gap greater than \(1/1000\). To our knowledge, this is the first fully analytic Bell-functional separation between qubit POVMs and qubit projective measurements over arbitrary shared two-qubit states. Lean certificate for the separation theorem is provided for completeness. Separately, an exact level-3 noncommutative sum-of-squares certificate proves that the explicit qubit strategy attains the unrestricted finite-dimensional tensor-product quantum optimum. 
\end{abstract}

\maketitle

\section{Introduction}
\label{sec:introduction}

Bell inequalities distinguish correlations compatible with local hidden-variable models from those allowed by quantum theory~\cite{Bell1964,CHSH1969}. General quantum measurements are described by positive-operator-valued measures (POVMs), whereas projection-valued measures (PVMs) have mutually orthogonal projection effects. Naimark's theorem realizes every POVM as a PVM on an enlarged Hilbert space~\cite{Beneduci2020}. The distinction therefore becomes operational when the local dimension, or the auxiliary quantum resources available for measurement, is constrained. This raises the question of whether generalized measurements can generate Bell correlations that projective measurements cannot reproduce within the same fixed dimension.

\begin{figure}[t]
    \centering
    \includegraphics[width=\linewidth]{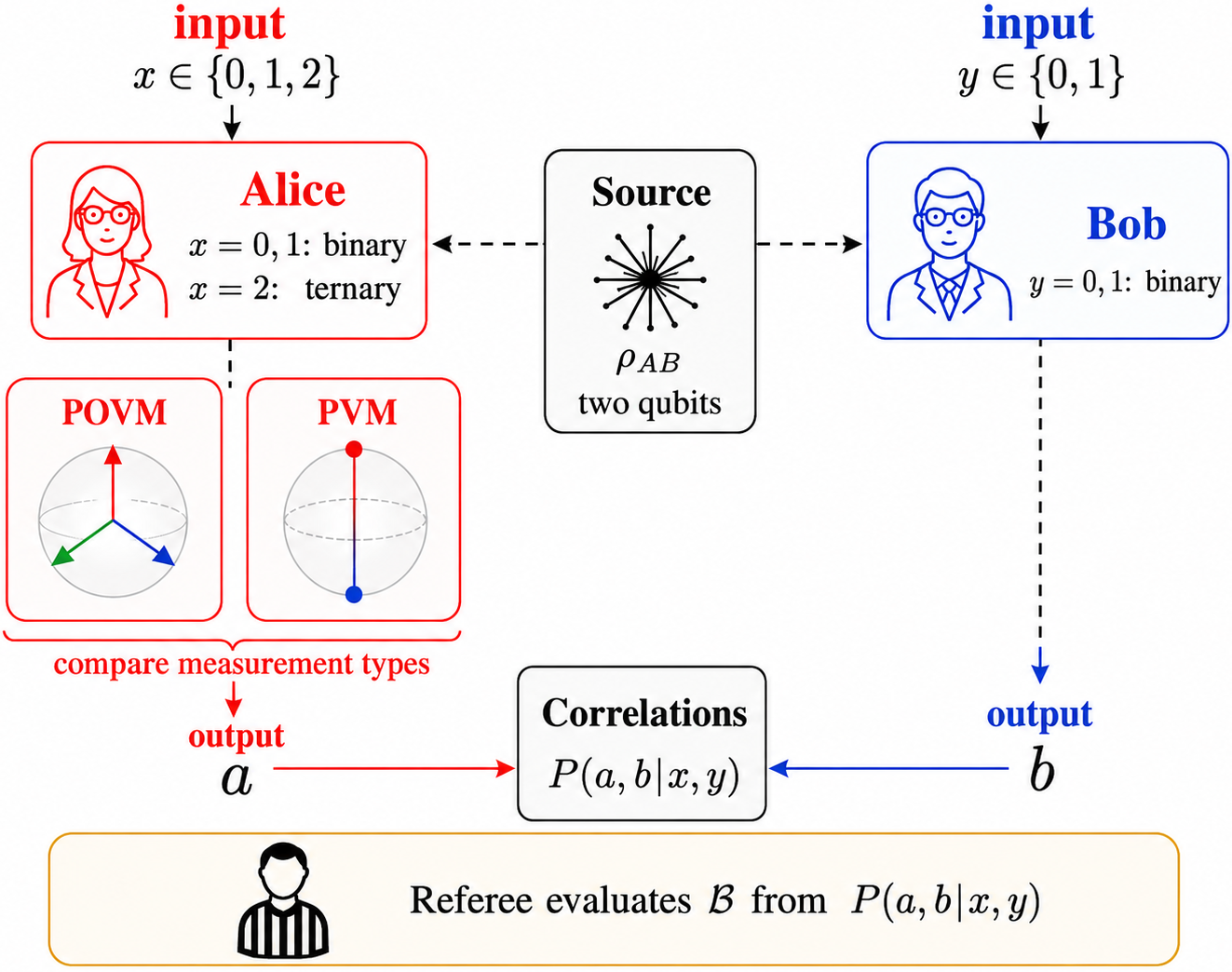}
    \caption{Bell scenario for comparing qubit POVMs with qubit projective measurements. A source distributes a two-qubit state \(\rho_{AB}\) to Alice and Bob. Alice receives \(x\in\{0,1,2\}\), where \(x=0,1\) label binary measurements and \(x=2\) labels a ternary measurement. Bob receives \(y\in\{0,1\}\) and performs a binary measurement. Their outputs determine conditional probabilities \(P(a,b\mid x,y)\), which are evaluated by the Bell functional \(\mathcal B\). In the separating strategy, Alice's ternary measurement is a genuine qubit POVM, whereas every measurement in the fixed-qubit comparison class is projective.}
    \label{fig:bell-game-setup}
\end{figure}

Gisin highlighted the absence, at the time, of a Bell inequality whose optimal quantum violation required a POVM~\cite{Gisin2009}. P\'al and V\'ertesi later showed that finite-dimensional quantum correlation sets can be nonconvex and exhibited correlations attainable with nonprojective qubit effects but not with a single qubit-projective realization~\cite{PalVertesi2009}. This established a set-theoretic distinction at fixed dimension, but not a linear separation from the convex hull of projective strategies, which is needed to exclude simulations using shared classical randomness.

V\'ertesi and Bene introduced an asymmetric Bell scenario in which Alice has two binary settings and one ternary setting, while Bob has two binary settings~\cite{VertesiBene2010}. Their Bell functional combines a dominant binary term with a weaker ternary contribution favoring a genuine three-outcome qubit POVM. They proved analytically a POVM advantage for a fixed maximally entangled two-qubit state and supported the optimization over arbitrary two-qubit states numerically. Barra \emph{et al.} later studied a broader parametrization of the ternary measurement and its detection efficiency~\cite{BarraEtAl2012}. Complementary work formalized projective simulability using PVMs, classical randomization, and classical postprocessing~\cite{OszmaniecEtAl2017}, while more recent studies considered projectivity under full-Schmidt-rank assumptions and perturbative fixed-dimensional methods~\cite{BaptistaEtAl2025,CerfOllivier2026}.

These developments leave a precise analytic gap. A POVM advantage proved for a fixed maximally entangled state does not exclude a projective strategy that changes the shared state and sacrifices part of the dominant Bell score to improve the ternary contribution. Although this global two-qubit optimization was investigated numerically in Ref.~\cite{VertesiBene2010}, no closed analytic bound simultaneously covered arbitrary shared two-qubit states, degenerate binary measurements, and all rank patterns of a ternary qubit PVM. A Bell-functional separation from the convex hull of all qubit-projective behaviors over arbitrary two-qubit states, was therefore still missing. A separate question is whether the resulting qubit POVM is optimal only relative to qubit PVMs or among all finite-dimensional quantum strategies.

We address both questions in the scenario of Fig.~\ref{fig:bell-game-setup}. Retaining the setting--outcome architecture and anchor--probe structure of Ref.~\cite{VertesiBene2010}, we introduce a rational deformation of a three-ray qubit measurement and derive a global analytic stability bound. The dominant CHSH term anchors the shared state and Bob's observables. Its defect from the quantum maximum controls both the bias of Bob's reduced state and the nonorthogonality of his measurement directions. The CHSH score decreases linearly in this defect, whereas the projective ternary contribution can increase only as its square root. For a weak probe, the POVM advantage is therefore first order in the probe strength, while the optimal projective compensation is second order. This converts the geometric advantage on the exact CHSH-maximizing slice into a global arbitrary-state separation.

Our results have two components. First, we construct an explicit three-outcome qubit POVM and prove analytically that its Bell value exceeds that of every qubit-projective strategy with an arbitrary shared two-qubit state. The proof covers degenerate binary measurements, both ternary qubit-PVM rank patterns, shared classical randomness, and classical output processing. It therefore yields a strict separation from the convex hull of qubit-projective behaviors without a numerical semidefinite-programming hierarchy. A Lean~4 formalization of the separation theorem is provided for completeness. Second, an exact level-3 noncommutative sum-of-squares certificate proves that the explicit qubit strategy attains the unrestricted finite-dimensional tensor-product quantum optimum. Thus the same two-qubit realization is not only superior to every qubit PVM but also globally optimal when arbitrary finite local dimensions, states, and POVMs are allowed.


Since Naimark dilation removes the distinction between POVMs and PVMs only when enlarging the local Hilbert space is free, nonorthogonal effects can constitute a genuine resource under fixed-dimensional constraints. Although earlier Bell witnesses already demonstrated the relevance of nonprojective measurements, they did not provide a closed analytic separation over arbitrary two-qubit states and convex mixtures of projective strategies. Our contribution is an exact separation robust to state reoptimization, degenerate measurements, shared randomness, and classical postprocessing. Together with the Lean formalization and exact unrestricted-dimensional SOS certificate, this identifies the missing Naimark ancilla as an operational measurement resource.

\section{Scenario and main results}
\label{sec:scenario}

Consider a finite bipartite Bell scenario. For a fixed \(D\in\mathbb N\), a behavior is a family of conditional probabilities
\begin{equation}
 P(a,b\mid x,y)
 =\Tr\!\left[\rho\left(M_{a\mid x}\otimes N_{b\mid y}\right)\right],
\label{eq:behavior}
\end{equation}
where \(\rho\) acts on \(\mathbb C^{D_A}\otimes\mathbb C^{D_B}\), with \(D_A,D_B\leq D\), and the local effects obey
\begin{equation}
 M_{a\mid x}\succeq0,\quad \sum_aM_{a\mid x}=\id,
 \qquad
 N_{b\mid y}\succeq0,\quad \sum_bN_{b\mid y}=\id.
\end{equation}

Here, \(\id\) denotes the identity operator on the corresponding local
Hilbert space. Let \(\Q(D)\) be the set of all such behaviors, and let \(\Qp(D)\) be the subset obtained by requiring the effects of every measurement to be mutually orthogonal projections. Zero projections are allowed. The shared state is arbitrary in both classes. We write
\begin{equation}
 \Q=\bigcup_{D\geq1}\Q(D)
\label{eq:unrestrictedQ}
\end{equation}
for the finite-dimensional tensor-product quantum set. Since the functional below is linear and its optimum is attained, no closure issue affects the stated maximum.

Alice has settings \(x\in\{0,1,2\}\). The settings \(x=0,1\) have outcomes \(\alpha\in\{+1,-1\}\), while \(x=2\) has outcomes \(a\in\{0,1,2\}\). Bob has settings \(y\in\{0,1\}\), both with outcomes \(b\in\{+1,-1\}\). Define
\begin{align}
 E_{xy}
 &=\sum_{\alpha,b=\pm1}\alpha b\,p(\alpha,b\mid x,y),
 &&x=0,1, \label{eq:Edef}\\
 C_{ay}
 &=\sum_{b=\pm1}b\,p(a,b\mid2,y),
 &&a=0,1,2. \label{eq:Cdef}
\end{align}
The Bell functional is
\begin{align}
 S&=E_{00}+E_{01}+E_{10}-E_{11}, \label{eq:Sdef}\\
 F&=C_{00}-\frac35C_{10}+\frac45C_{11}
       -\frac35C_{20}-\frac45C_{21}, \label{eq:Fdef}\\
 \mathcal B&=S+\frac1{100}F. \label{eq:Bdef}
\end{align}
It is useful to collect the probe coefficients as
\begin{equation}
 \boldsymbol\nu_0=(1,0),\quad
 \boldsymbol\nu_1=\left(-\frac35,\frac45\right),\quad
 \boldsymbol\nu_2=\left(-\frac35,-\frac45\right).
\label{eq:nu}
\end{equation}
We denote the \(y\)th component of \(\boldsymbol\nu_a\) by \(\nu_{ay}\). Thus \(F=\sum_{a,y}\nu_{ay}C_{ay}\). All probability coefficients of \(\mathcal B\) are rational.

For a deterministic local strategy, the binary outputs can attain at most the local CHSH value \(2\). Independently, Alice's ternary output can be chosen to maximize \(\boldsymbol\nu_a\cdot(b_0,b_1)\), whose maximum over deterministic Bob outputs is \(7/5\). Hence the local bound is
\begin{equation}
 \beta_{\mathrm L}=2+\frac7{500}=\frac{1007}{500}=2.014.
\label{eq:local-bound}
\end{equation}
For a set of behaviors \(\mathcal S\), let \(\operatorname{conv}\mathcal S\) denote its convex hull.
\begin{theorem}[Analytic fixed-qubit POVM--PVM separation]
\label{thm:fixed-separation}
There exists an explicit behavior \(p_\star\in\Q(2)\), generated by a two-qubit state and a genuine three-outcome qubit POVM, such that
\begin{equation}
 \mathcal B(p_\star)=2\sqrt2+\frac1{100}.
\label{eq:povmvalue}
\end{equation}
For every shared state with local dimensions at most two and every family of local PVMs, including degenerate binary measurements and zero ternary projections,
\begin{equation}
 \mathcal B
 \leq2\sqrt2+\frac{\sqrt5}{250}
               +\frac{\sqrt2}{32400}.
\label{eq:pvmupper}
\end{equation}
Consequently,
\begin{equation}
 p_\star\notin\operatorname{conv}\Qp(2),
 \qquad
 \Qp(2)\subsetneq\Q(2),
\end{equation}
and the certified Bell-value gap obeys
\begin{equation}
 \frac1{100}-\frac{\sqrt5}{250}-\frac{\sqrt2}{32400}
 =\frac{1620-648\sqrt5-5\sqrt2}{162000}
 >\frac1{1000}.
\label{eq:gap}
\end{equation}
\end{theorem}

\begin{theorem}[Exact unrestricted finite-dimensional quantum optimum]
\label{thm:global-optimum}
For arbitrary finite local dimensions, arbitrary shared states, and arbitrary local POVMs in the scenario above,
\begin{equation}
 \max_{p\in\Q}\mathcal B(p)
 =2\sqrt2+\frac1{100}.
\label{eq:globalvalue}
\end{equation}
The upper bound is established by an exact level-3 noncommutative sum-of-squares (SOS) certificate over \(\Ksqrt\), where $\Ksqrt:=\{a+b\sqrt2:a,b\in\mathbb Q\}$.
\end{theorem}


For comparison, Table~\ref{tab:benchmarks} summarizes the local bound, the analytic qubit-PVM upper bound, the value attained by the explicit qubit POVM, and the exact unrestricted finite-dimensional quantum optimum.

\begin{table}[h]
\caption{Relevant Bell-value benchmarks. The qubit-PVM entry is a rigorous analytic upper bound, not a claimed exact optimum.}
\label{tab:benchmarks}
\begin{ruledtabular}
\begin{tabular}{lc}
Class & Bell value \\
\hline
Local behaviors & \(2.014\) \\
Qubit PVMs & \(\leq 2.8374150452\) \\
Explicit qubit POVM & \(2.8384271247\) \\
All finite-dimensional POVMs & \(2.8384271247\) \\
\end{tabular}
\end{ruledtabular}
\end{table}

Theorem~\ref{thm:fixed-separation} is proved entirely by the explicit construction and the dimension-specific geometry in Secs.~\ref{sec:povm} and \ref{sec:pvm-proof}. Although the proof is human-readable, for completeness we also provide a Lean~4 formalization based on Lean-QIT~\cite{zhu2026leanqit}, available in the \href{https://github.com/real-lin-zhu/Exact-Fixed-Dimension-Bell-Separation-between-Qubit-POVMs-and-PVMs/tree/main/Lean_formalisation}{\texttt{Lean\_formalisation}} directory of the \href{https://github.com/real-lin-zhu/Exact-Fixed-Dimension-Bell-Separation-between-Qubit-POVMs-and-PVMs}{public GitHub repository}. Theorem~\ref{thm:global-optimum} is a separate exact computer-assisted statement proved in Sec.~\ref{sec:sos} and Appendix~\ref{app:exact-sos-proof}.

\section{Geometric design and intuition}
\label{sec:intuition}

Let
\begin{equation}
 X=\begin{pmatrix}0&1\\1&0\end{pmatrix},
 \qquad
 Z=\begin{pmatrix}1&0\\0&-1\end{pmatrix}
\label{eq:pauli-XZ}
\end{equation}
be the Pauli \(X\) and \(Z\) operators.

Consider first the one-parameter family
\begin{equation}
 \mathcal{B}_\varepsilon=S+\varepsilon F,
 \qquad 0<\varepsilon\ll1,
\label{eq:Bepsilon}
\end{equation}
of which Eq.~\eqref{eq:Bdef} is the case \(\varepsilon=1/100\).  The CHSH expression \(S\) is the anchor, and \(F\) is a smaller probe of Alice's ternary measurement.

At the exact CHSH maximum, one may use the maximally entangled state \(\ket{\Phi^+}=(\ket{00}+\ket{11})/\sqrt2\) and identify Bob's two observables with \(Z\) and \(X\).  Put
\begin{equation}
 G_a=\nu_{a0}Z+\nu_{a1}X.
\end{equation}
Every \(\boldsymbol{\nu}_a\) in Eq.~\eqref{eq:nu} is a unit vector, so \(G_a\) has eigenvalues \(+1\) and \(-1\).  For a ternary measurement \(\{M_a\}\) on Alice's qubit, the real maximally entangled identity gives
\begin{align}
 F
 &=\frac12\sum_a\Tr(M_a^{\mathsf T}G_a)
  =\frac12\sum_a\Tr(M_aG_a) \notag\\
 &\leq\frac12\sum_a\Tr M_a=1.
\label{eq:Fsliceupper}
\end{align}
The upper bound is attained when the support of \(M_a\) is the \(+1\) eigenspace of \(G_a\).

The positive weights
\begin{equation}
 w_0=\frac34,\qquad w_1=w_2=\frac58
\label{eq:weights}
\end{equation}
satisfy
\begin{equation}
 \sum_aw_a=2,\qquad
 \sum_aw_a\boldsymbol{\nu}_a=0.
\label{eq:balance}
\end{equation}
The three effects
\begin{equation}
 M_a=\frac{w_a}{2}(\id+G_a)
\label{eq:aligned-effects}
\end{equation}
therefore form a POVM and saturate Eq.~\eqref{eq:Fsliceupper}.  Geometrically, the three weighted rays surround the origin.  The \(3\)-\(4\)-\(5\) coordinates are a rational deformation of a trine and make both the Bell coefficients and the effect matrices rational.

A ternary PVM on a qubit has, up to relabeling, rank pattern \((2,0,0)\) or \((1,1,0)\).  On the exact CHSH slice the former gives \(F=0\), because Bob's marginals vanish.  In the latter case the two rank-one projectors form one antipodal Bloch pair.  For a chosen pair of outcome labels \(i,j\), optimizing its Bloch axis gives
\begin{equation}
 \left\|\frac{\boldsymbol{\nu}_i-\boldsymbol{\nu}_j}{2}\right\|_2.
\end{equation}
Let \(F_{\mathrm{PVM}}\) and \(F_{\mathrm{POVM}}\) denote the maximal values of the probe \(F\) on this exact CHSH-maximizing slice when Alice's ternary measurement is restricted to a qubit PVM or allowed to be a general qubit POVM, respectively. It follows that
\begin{equation}
 F_{\mathrm{PVM}}\leq
 \frac12\max_{i<j}\norm{\boldsymbol{\nu}_i-\boldsymbol{\nu}_j}_2
 =\frac2{\sqrt5},
 \qquad
 F_{\mathrm{POVM}}=1.
\label{eq:slicegap}
\end{equation}
Thus the exact-slice gap is \(1-2/\sqrt5\): a qubit POVM can align three nonantipodal effects, while a qubit PVM is confined to one diameter.

The remaining issue is global.  A projective strategy may sacrifice part of its CHSH score to increase \(F\).  For \(S>2\), introduce the defect
\begin{equation}
 q=2-\frac{S^2}{4}.
\label{eq:qintro}
\end{equation}
The proof below shows that \(q\) controls both the length of Bob's marginal Bloch vector and the nonorthogonality of his measurement directions, and consequently
\begin{equation}
 F_{\mathrm{PVM}}\leq\frac2{\sqrt5}+\frac{10}{9}\sqrt q.
\label{eq:Fdefectpreview}
\end{equation}
Since \(S=2\sqrt{2-q}\), the anchor loses at order \(q\) while the probe can gain only at order \(\varepsilon\sqrt q\).  Optimizing this tradeoff gives
\begin{align}
 \mathcal{B}_\varepsilon
 &\leq\varepsilon\frac2{\sqrt5}
   +2\sqrt{2-q}+\frac{10\varepsilon}{9}\sqrt q \notag\\
 &\leq\varepsilon\frac2{\sqrt5}
   +\sqrt{8+\frac{200}{81}\varepsilon^2} \notag\\
 &\leq2\sqrt2+\varepsilon\frac2{\sqrt5}
   +\frac{25\sqrt2}{81}\varepsilon^2.
\label{eq:tradeoffpreview}
\end{align}
By contrast, the aligned POVM attains \(2\sqrt2+\varepsilon\). The first-order probe advantage therefore dominates the second-order compensation available to projective strategies for sufficiently small \(\varepsilon\).  The rational choice \(\varepsilon=1/100\) yields the explicit constants in Theorem~\ref{thm:fixed-separation}.

\section{Explicit POVM strategy}
\label{sec:povm}

Using the Pauli operators defined in Eq.~\eqref{eq:pauli-XZ}, the
parties share \(\ket{\Phi^+}\).  Bob uses
\begin{equation}
 B_0=Z,\qquad B_1=X,
\end{equation}
and Alice's binary observables are
\begin{equation}
 A_0=\frac{Z+X}{\sqrt2},\qquad
 A_1=\frac{Z-X}{\sqrt2}.
\end{equation}
The effects for outcome \(\pm1\) are
\((\id\pm A_x)/2\) and \((\id\pm B_y)/2\), respectively.

For Alice's ternary setting, Eq.~\eqref{eq:aligned-effects} becomes
\begin{align}
 M_0&=\begin{pmatrix}3/4&0\\0&0\end{pmatrix},\qquad
 M_1=\begin{pmatrix}1/8&1/4\\1/4&1/2\end{pmatrix}, \notag\\
 M_2&=\begin{pmatrix}1/8&-1/4\\-1/4&1/2\end{pmatrix}.
\label{eq:povm}
\end{align}
Each effect is positive semidefinite of rank one, and Eq.~\eqref{eq:balance} gives \(\sum_aM_a=\id\).  Their nonzero eigenvalues are \(3/4,5/8,5/8\), so none is a projection.

For arbitrary qubit operators \(R\) and \(T\), using
\begin{equation}
 \bra{\Phi^+}R\otimes T\ket{\Phi^+}
 =\frac12\Tr(R^{\mathsf T}T),
\end{equation}
where the transpose is taken in the computational basis used to define
\(\ket{\Phi^+}\), one obtains
\begin{equation}
 S=2\sqrt2,\qquad
 C_{ay}=\frac{w_a}{2}\nu_{ay}.
\end{equation}
Consequently,
\begin{equation}
 F=\frac12\sum_aw_a\norm{\boldsymbol{\nu}_a}_2^2=1
\end{equation}
and \(\mathcal{B}=2\sqrt2+1/100\), which proves the achievability statement in Theorem~\ref{thm:fixed-separation}.  The matching dimension-unrestricted upper bound is proved separately in Sec.~\ref{sec:sos}.

\begin{proposition}
\label{prop:genuine}
The ternary measurement in Eq.~\eqref{eq:povm} is neither a stochastic postprocessing of a single qubit PVM nor a convex mixture of qubit PVM-valued measurements with the same three outcomes.
\end{proposition}

\begin{proof}
The first claim follows from
\begin{equation}
 [M_0,M_1]
 =\begin{pmatrix}0&3/16\\-3/16&0\end{pmatrix}\neq0,
\end{equation}
whereas all effects obtained by postprocessing one PVM commute.

For the second claim, first note that the three effects are linearly independent.  Indeed, the off-diagonal and lower-right entries of \(\sum_ac_aM_a=0\) imply \(c_1=c_2=0\), and the upper-left entry then gives \(c_0=0\).  Suppose \(M_a=tN_a+(1-t)L_a\), with \(0<t<1\), for two POVMs \(\{N_a\}\) and \(\{L_a\}\).  Positivity and the rank-one support of \(M_a\) force both \(N_a\) and \(L_a\) to be supported on the range of \(M_a\).  Hence \(N_a-M_a=c_aM_a\) for some real \(c_a\). Completeness gives \(\sum_ac_aM_a=0\), and linear independence gives \(c_a=0\) for every \(a\).  Thus the POVM is extremal.  If it were a nontrivial convex mixture of PVM-valued POVMs, extremality would force every component to equal it, which is impossible because its effects are not projections.
\end{proof}

\section{Analytic upper bound for qubit PVM strategies}
\label{sec:pvm-proof}

\subsection{Initial reductions}

Fix arbitrary projective measurements and an arbitrary shared state.  The Bell score is affine in the state, so a pure-state decomposition reduces the proof to an arbitrary pure state \(\ket{\psi}\), with the measurements held fixed.

For the binary settings, denote Alice's and Bob's projective effects by
\begin{equation}
 P_{\alpha\mid x}:=M_{\alpha\mid x},
 \qquad
 Q_{b\mid y}:=N_{b\mid y}.
\end{equation}
Define the corresponding observables by
\begin{equation}
 A_x=P_{+1\mid x}-P_{-1\mid x},
 \qquad
 B_y=Q_{+1\mid y}-Q_{-1\mid y}.
\end{equation}
These observables are self-adjoint unitaries, including the degenerate cases \(\pm\id\).  Tsirelson's bound gives \(S\leq2\sqrt2\) ~\cite{Tsirelson1980}.  We abbreviate
\begin{equation}
 u=\frac2{\sqrt5}.
\label{eq:u}
\end{equation}

Let \(P_a=P_{a\mid2}\) denote Alice's ternary PVM and define
\begin{equation}
 p_A(a\mid2)=\bra{\psi}P_a\otimes\id\ket{\psi}.
\end{equation}
Since \(-\id\preceq B_y\preceq\id\),
\begin{equation}
 \abs{C_{ay}}\leq p_A(a\mid2).
\end{equation}
Therefore
\begin{equation}
 F\leq\sum_ap_A(a\mid2)\norm{\boldsymbol{\nu}_a}_1
 \leq\max_a\norm{\boldsymbol{\nu}_a}_1=\frac75.
\label{eq:Fcoarse}
\end{equation}
If \(S\leq2\), then \(\mathcal{B}\leq2+7/500<2\sqrt2\), which is already below Eq.~\eqref{eq:pvmupper}.  It remains to consider \(S>2\).

\subsection{One defect controls the anchored geometry}
\label{sec:defect}

If any binary qubit PVM is degenerate, its observable is \(\pm\id\) and commutes with the other observable on the same side.  Define the CHSH operator
\begin{equation}
 \mathsf C
 =
 A_0\otimes(B_0+B_1)
 +
 A_1\otimes(B_0-B_1).
\label{eq:CHSH-operator}
\end{equation}
It satisfies
\begin{equation}
 \mathsf C^2
 =
 4\id-[A_0,A_1]\otimes[B_0,B_1],
\end{equation}
so in that case \(\norm{\mathsf{C}}=2\), contradicting \(S>2\). All four binary observables are therefore nondegenerate.  In particular, writing
\(\boldsymbol{\sigma}:=(X,Y,Z)\) for the vector of Pauli matrices, Bob's
nondegenerate binary observables can be expressed as
\begin{equation}
 B_y=\boldsymbol{b}_y\cdot\boldsymbol{\sigma},
 \qquad
 t:=\boldsymbol{b}_0\cdot\boldsymbol{b}_1\in[-1,1],
\end{equation}
where \(\boldsymbol{b}_0,\boldsymbol{b}_1\in\mathbb{R}^3\) are unit Bloch vectors.
Define
\begin{equation}
 q=2-\frac{S^2}{4}\in[0,1).
\label{eq:q}
\end{equation}
Using only \(\norm{A_x}=1\), we have
\begin{align}
 S
 &\leq\norm{B_0+B_1}+\norm{B_0-B_1} \notag\\
 &=\sqrt{2+2t}+\sqrt{2-2t}.
\end{align}
Squaring and using Eq.~\eqref{eq:q} gives \(\sqrt{1-t^2}\geq1-q\).  A second squaring yields
\begin{equation}
 t^2\leq2q-q^2\leq2q.
\label{eq:tq}
\end{equation}

Let \(r\) be the length of the Bloch vector of Bob's reduced state.  Up to local unitaries, every pure two-qubit state has the Schmidt form
\begin{equation}
 \ket{\psi}=\cos\theta\ket{00}+\sin\theta\ket{11}.
\end{equation}
Its correlation matrix has singular values \(1,s,s\), where \(s=\sin(2\theta)=\sqrt{1-r^2}\).  The two-qubit CHSH criterion of the Horodeckis~\cite{Horodecki1995}, or its direct singular-value derivation, then gives
\begin{equation}
 S\leq2\sqrt{1+s^2}=2\sqrt{2-r^2}.
\end{equation}
Hence
\begin{equation}
 r^2\leq q.
\label{eq:rq}
\end{equation}
Equations~\eqref{eq:tq} and \eqref{eq:rq} are the only near-maximal-CHSH properties required below.

\subsection{Ternary projective rank patterns}
\label{sec:rankpatterns}

Three mutually orthogonal projections summing to the identity on \(\mathbb{C}^2\) have, up to output permutation, rank pattern \((2,0,0)\) or \((1,1,0)\). For a Bob-side operator \(O\), write
\begin{equation}
 \langle O\rangle
 :=
 \bra{\psi}\id\otimes O\ket{\psi}.
\end{equation}

For rank pattern \((2,0,0)\), let \(P_i=\id\) be the nonzero projection.  Since \(B_y\) is a traceless Pauli observable, \(\abs{\langle B_y\rangle}\leq r\), and therefore
\begin{equation}
 F=\sum_y\nu_{iy}\langle B_y\rangle
 \leq\frac75r
 \leq\frac75\sqrt q.
\label{eq:Frank200}
\end{equation}

For rank pattern \((1,1,0)\), let \(P_i,P_j\) be the two nonzero rank-one projections, and set
\begin{equation}
 A_2=P_i-P_j,\qquad
 \boldsymbol{s}=\frac{\boldsymbol{\nu}_i+\boldsymbol{\nu}_j}{2},
 \qquad
 \boldsymbol{d}=\frac{\boldsymbol{\nu}_i-\boldsymbol{\nu}_j}{2}.
\end{equation}
Then
\begin{equation}
 F=\sum_ys_y\langle B_y\rangle+
 \bra{\psi}A_2\otimes(d_0B_0+d_1B_1)\ket{\psi}.
\label{eq:Fdecomp}
\end{equation}
For each of the three unordered pairs \((i,j)\),
\(\norm{\boldsymbol{s}}_1=3/5\).  Moreover,
\begin{equation}
 (d_0B_0+d_1B_1)^2
 =(d_0^2+d_1^2+2d_0d_1t)\id,
\end{equation}
and direct evaluation of the three pairs gives
\begin{equation}
 \max_{i<j}\norm{d_0B_0+d_1B_1}^2
 =\frac45+\frac{16}{25}\abs t.
\end{equation}
Using \(\sqrt{a+x}\leq\sqrt a+x/(2\sqrt a)\) for \(a>0\) and
\(x\geq0\), we obtain
\begin{equation}
 \norm{d_0B_0+d_1B_1}
 \leq u+\frac{4\sqrt5}{25}\abs t.
\end{equation}
Equations~\eqref{eq:Fdecomp}, \eqref{eq:tq}, and \eqref{eq:rq} now imply
\begin{align}
 F
 &\leq u+\frac{4\sqrt5}{25}\sqrt{2q}+\frac35\sqrt q \notag\\
 &=u+\left(\frac35+\frac{4\sqrt{10}}{25}\right)\sqrt q.
\label{eq:FwithK}
\end{align}
Set
\begin{equation}
 K=\frac35+\frac{4\sqrt{10}}{25}.
\end{equation}
The exact inequality \(K<10/9\) is equivalent to \(36\sqrt{10}<115\), which follows after squaring from \(12960<13225\).  Thus
\begin{equation}
 F\leq u+\frac{10}{9}\sqrt q
\label{eq:F111}
\end{equation}
for rank pattern \((1,1,0)\).

The same bound covers Eq.~\eqref{eq:Frank200}.  Indeed, because
\(\sqrt q<1\) and \(u>4/5\),
\begin{equation}
 \frac75\sqrt q
 \leq\frac35\sqrt q+\frac45
 <u+\frac35\sqrt q
 \leq u+\frac{10}{9}\sqrt q.
\end{equation}
Every qubit-PVM strategy with \(S>2\) therefore satisfies
\begin{align}
 \mathcal{B}
 &\leq2\sqrt{2-q}
 +\frac1{100}\left(u+\frac{10}{9}\sqrt q\right) \notag\\
 &=\frac{u}{100}+2\sqrt{2-q}+\frac1{90}\sqrt q.
\label{eq:BbeforeCS}
\end{align}
The Cauchy--Schwarz inequality gives
\begin{equation}
 2\sqrt{2-q}+\frac1{90}\sqrt q
 \leq\sqrt{8+\frac1{4050}}.
\end{equation}
Applying the tangent-line bound for the square root once more,
\begin{equation}
 \sqrt{8+\frac1{4050}}
 \leq2\sqrt2+\frac{\sqrt2}{32400}.
\end{equation}
Since \(u/100=\sqrt5/250\), this proves the first inequality in
Eq.~\eqref{eq:pvmupper}.  Finally,
\begin{equation}
 \frac{\sqrt5}{250}+\frac{\sqrt2}{32400}
 <\frac{179}{20000}+\frac1{20000}
 =\frac9{1000},
\end{equation}
where we used
\(\sqrt5<179/80\) and \(\sqrt2<81/50\).  This proves Eq.~\eqref{eq:pvmupper} for every pure state and, by averaging, for every mixed state.  If either local dimension is one, \(S\leq2\), so that case was already covered by Eq.~\eqref{eq:Fcoarse}.  Together with the explicit strategy of Sec.~\ref{sec:povm}, this completes the proof of Theorem~\ref{thm:fixed-separation}. 

\subsection{A separating interval of anchor--probe witnesses}
\label{sec:epsilon-family}

The same estimates establish a family of fixed-qubit separations rather than only the rational instance \(\varepsilon=1/100\). Retain the sharper constant
\begin{equation}
 K=\frac35+\frac{4\sqrt{10}}{25},
 \qquad
 u=\frac2{\sqrt5}.
\label{eq:epsilon-family-constants}
\end{equation}
For \(S\leq2\), Eq.~\eqref{eq:Fcoarse} gives
\begin{equation}
 \mathcal B_\varepsilon\leq2+\frac75\varepsilon.
\label{eq:epsilon-coarse-branch}
\end{equation}
For \(S>2\), Eqs.~\eqref{eq:FwithK} and \eqref{eq:q} imply
\begin{align}
 \mathcal B_\varepsilon
 &\leq u\varepsilon+2\sqrt{2-q}+K\varepsilon\sqrt q \\ 
 &\leq u\varepsilon+\sqrt{8+2K^2\varepsilon^2},
\label{eq:epsilon-pvm-bound}
\end{align}
where the second inequality is Cauchy--Schwarz. The explicit qubit POVM of Sec.~\ref{sec:povm} attains \(2\sqrt2+\varepsilon\).

\begin{corollary}[Interval of analytic fixed-qubit separations]
\label{cor:epsilon-family}
For every
\begin{equation}
 0<\varepsilon<\varepsilon_\star,
 \qquad
 \varepsilon_\star=
 \frac{4\sqrt2(1-u)}{2K^2-(1-u)^2}
 \approx0.24524,
\label{eq:epsilon-threshold}
\end{equation}
the Bell functional \(\mathcal B_\varepsilon=S+\varepsilon F\) separates the explicit qubit POVM strategy from \(\operatorname{conv}\Qp(2)\).
\end{corollary}

\begin{proof}
The branch in Eq.~\eqref{eq:epsilon-coarse-branch} is below \(2\sqrt2+\varepsilon\) whenever \(\varepsilon<5(\sqrt2-1)\). For the branch in Eq.~\eqref{eq:epsilon-pvm-bound}, the desired strict inequality is equivalent, after moving \(u\varepsilon\) to the right and squaring two positive quantities, to
\begin{equation}
 \left[2K^2-(1-u)^2\right]\varepsilon
 <4\sqrt2(1-u).
\end{equation}
The threshold in Eq.~\eqref{eq:epsilon-threshold} is smaller than \(5(\sqrt2-1)\), so it controls both branches.
\end{proof}

\section{Exact dimension-unrestricted quantum optimum}
\label{sec:sos}

The explicit strategy of Sec.~\ref{sec:povm} supplies the lower bound in Theorem~\ref{thm:global-optimum}. We now establish the matching upper bound for arbitrary finite local dimensions. This result is logically independent of Theorem~\ref{thm:fixed-separation}: the fixed-qubit separation is human-readable and analytic, whereas the dimension-unrestricted statement is computer-assisted but verified by an exact algebraic certificate.

For projector coordinates, write
\begin{equation}
 P_x=M_{+1\mid x},\qquad
 Q_y=N_{+1\mid y},\qquad
 R_a=M_{a\mid2},
\end{equation}
where \(x,y=0,1\), \(a=0,1,2\), and
\(R_2=\id-R_0-R_1\).  The generators
\(P_0,P_1,R_0,R_1,Q_0,Q_1\) are self-adjoint and obey
\begin{align}
 P_x^2&=P_x, & Q_y^2&=Q_y,\notag\\
 R_a^2&=R_a, & R_0R_1&=R_1R_0=0,\notag\\
 [P_x,Q_y]&=0, & [R_a,Q_y]&=0
 \quad(a=0,1).
\label{eq:projector-algebra}
\end{align}
No commutation relation is imposed between measurements belonging to different settings on the same party.

Let \(\widehat A_x=2P_x-\id\) and \(\widehat B_y=2Q_y-\id\).  Suppressing tensor-product symbols, the Bell operator corresponding to Eq.~\eqref{eq:Bdef} is
\begin{align}
 \Bellop={}&
 \widehat A_0\widehat B_0+\widehat A_0\widehat B_1
 \notag\\
 &+\widehat A_1\widehat B_0-\widehat A_1\widehat B_1
 \notag\\
 &+\frac1{100}\bigg(
 R_0\widehat B_0-\frac35R_1\widehat B_0
 +\frac45R_1\widehat B_1\bigg)
 \notag\\
 &-\frac1{100}\bigg(
 \frac35R_2\widehat B_0+\frac45R_2\widehat B_1
 \bigg).
\label{eq:Belloperator}
\end{align}
Equivalently, after eliminating \(R_2\),
\begin{align}
 \Bellop={}&
 \frac{1007}{500}\id-4P_0
 -\frac{1003}{250}Q_0-\frac2{125}Q_1
 \notag\\
 &+4P_0Q_0+4P_0Q_1
 +4P_1Q_0-4P_1Q_1
 \notag\\
 &-\frac3{125}R_0-\frac2{125}R_1
 +\frac4{125}R_0Q_0+\frac2{125}R_0Q_1
 \notag\\
 &+\frac4{125}R_1Q_1 .
\label{eq:Belloperator-projectors}
\end{align}

The explicit qubit strategy of Sec.~\ref{sec:povm} attains the Bell value \(\Omega:=2\sqrt2+1/100\).  To prove that this value is the dimension-unrestricted quantum maximum, it remains to show that \(\Bellop\leq\Omega\id\) in every representation of the projector algebra in Eq.~\eqref{eq:projector-algebra}.  The following exact level-3 sum-of-squares certificate establishes this operator inequality.

\begin{proposition}[Exact level-3 SOS certificate]
\label{prop:exact-sos}
Apply the deterministic reduction convention specified in Appendix~\ref{app:exact-sos-proof} to words in the quotient \(*\)-algebra defined by Eq.~\eqref{eq:projector-algebra}. There exist an ordered column vector \(\boldsymbol v\) containing the \(83\) distinct reduced words of degree at most three used by the certificate and a real symmetric matrix \(\Gram\in\Ksqrt^{83\times83}\) such that
\begin{equation}
    \Omega\id-\Bellop
    =
    \boldsymbol v^\dagger\Gram\boldsymbol v,
    \qquad
    \Gram\succeq0.
    \label{eq:exact-sos}
\end{equation}
\end{proposition}

\begin{proof}[Proof sketch]
Applying idempotence, orthogonality, and Alice--Bob commutation gives an ordered list of \(83\) distinct reduced words of degree at most three.  Expanding
\(\boldsymbol v^\dagger\Gram\boldsymbol v\) with respect to the fixed reduced-word list and
matching coefficients over \(\Ksqrt\) proves the polynomial identity in
Eq.~\eqref{eq:exact-sos}.

It remains to establish positivity without numerical approximation.  The
exact Gram matrix admits a factorization
\(\Gram=WHW^{\mathsf T}\), where \(W\) has full column rank.  Moreover,
there exists an invertible rational matrix \(U\) such that
\(\widetilde H=U^{\mathsf T}HU\) has positive diagonal entries and is
strictly diagonally dominant over \(\Ksqrt\).  Define
\begin{equation}
    \delta_i
    :=
    \widetilde H_{ii}
    -
    \sum_{j\neq i}
    \abs{\widetilde H_{ij}}
    >0.
    \label{eq:sos-diagonal-dominance}
\end{equation}
Let \(e_i\) denote the \(i\)th standard basis vector. With \(s_{ij}:=\operatorname{sgn}(\widetilde H_{ij})\), one has
\begin{align}
    \widetilde H
    ={}&
    \sum_i
    \delta_i e_i e_i^{\mathsf T}
    \notag\\
    &+
    \sum_{i<j}
    \abs{\widetilde H_{ij}}
    (e_i+s_{ij}e_j)
    (e_i+s_{ij}e_j)^{\mathsf T}.
    \label{eq:sos-diagonal-decomposition}
\end{align}
Hence \(\widetilde H\succ0\).  Since \(U\) is invertible, it follows that \(H\succ0\), and therefore \(\Gram\succeq0\).  The exact verification of the polynomial identity and the positivity of \(\Gram\) is given in Appendix~\ref{app:exact-sos-proof}.
\end{proof}

For any representation satisfying Eq.~\eqref{eq:projector-algebra}, the right-hand side of Eq.~\eqref{eq:exact-sos} is positive.  Hence \(\langle\Bellop\rangle\leq\Omega\) for every finite-dimensional tensor-product projective strategy.

To extend this projective upper bound to arbitrary finite families of local POVMs, we apply the simultaneous finite-family form of Naimark dilation \cite{BaptistaEtAl2025}, independently to Alice and Bob. The fixed local embeddings preserve every joint probability, while the setting-dependent PVMs satisfy the projector relations used in the SOS identity. Hence the same upper bound holds for every \(p\in\Q\). Together with the explicit qubit strategy of Sec.~\ref{sec:povm}, this proves Theorem~\ref{thm:global-optimum}.



\section{Discussion}
\label{sec:discussion}

The fixed-qubit separation is global over all two-qubit states and all allowed PVMs, not only over strategies on the exact CHSH-maximizing slice. Because a linear functional has the same supremum on a set and on its convex hull, the upper bound also applies when projective strategies are supplemented with shared classical randomness. A finite stochastic output map is a convex combination of deterministic maps, and deterministic relabeling or coarse-graining of a PVM again gives a PVM, possibly with zero projections. Thus neither shared classical randomness nor classical output processing closes the gap at local dimension two; see also the general study of projective simulability in Ref.~\cite{OszmaniecEtAl2017}.

The operational statement must nevertheless be interpreted with the dimension assumption visible. The observed Bell value certifies that Alice's ternary measurement cannot be simulated using qubit PVMs, arbitrary two-qubit states, shared randomness, and classical output processing. It does not exclude a projective implementation after introducing a quantum ancilla, because Naimark dilation enlarges the local Hilbert space. The resource witnessed by the theorem is therefore access to nonorthogonal measurement effects within a fixed two-dimensional laboratory, or equivalently access to the effective measurement space supplied by an ancilla.

The exact SOS certificate adds a logically different dimension-unrestricted statement. Even after arbitrary finite local dimensions and general POVMs are allowed, the Bell score cannot exceed \(2\sqrt2+1/100\). Thus the explicit qubit POVM is globally optimal. This does not prove that the maximizing realization is unique, does not self-test the POVM, and does not make a POVM necessary once larger local dimensions are permitted.

Relative to the fixed-state analytic comparison and numerical arbitrary-state optimization of V\'ertesi and Bene~\cite{VertesiBene2010}, the present construction converts the global qubit-PVM exclusion into a closed analytic argument. Its central mechanism is an anchor--probe tradeoff. One defect \(q\) controls both the reduced-state bias and the overlap of Bob's measurement directions. The restricted projective class can gain only \(O(\sqrt q)\) in the probe while losing \(O(q)\) in CHSH, so a probe of strength \(\varepsilon\) produces a first-order POVM advantage and only a second-order projective compensation. Corollary~\ref{cor:epsilon-family} shows that this mechanism yields an interval of separating witnesses rather than one isolated rational point. The value \(\varepsilon=1/100\) is distinguished because it also admits the exact unrestricted SOS certificate presented here.

The present witness is mathematically exact but has limited experimental
robustness. Consider the uniform-white-noise model
\begin{equation}
    p_v
    =
    v p_\star+(1-v)p_{\mathrm{unif}},
    \qquad
    0\leq v\leq1,
\end{equation}
where \(p_{\mathrm{unif}}\) denotes the uniform behavior and satisfies
\(\mathcal B(p_{\mathrm{unif}})=0\). To exclude all qubit-PVM strategies
using the analytic upper bound in Eq.~\eqref{eq:pvmupper}, the visibility
must satisfy
\begin{equation}
    v
    >
    \frac{
        2\sqrt2+\sqrt5/250+\sqrt2/32400
    }{
        2\sqrt2+1/100
    }
    \approx 0.9996434.
    \label{eq:white-noise-visibility}
\end{equation}
This visibility should not be interpreted as a finite-statistics threshold
or as a robust self-testing guarantee, neither of which is established
here. Improving the noise tolerance by optimizing the probe directions and
the probe strength remains an important direction beyond the exact
separation proved in this work.

The scope of the result is precise. The analytic theorem does not determine the exact optimum over \(\Qp(2)\), although it gives a sharp closed upper bound. The construction uses three settings for Alice and therefore does not settle the exactly-two-settings-per-party program of Ref.~\cite{CerfOllivier2026}. It also does not provide uniqueness, measurement reconstruction, or a commuting-operator strengthening. These boundaries leave open the exact qubit-PVM optimum, more robust witnesses, general anchor--probe separation theorems, and analogous fixed-dimension separations under tighter setting constraints.

\section{Conclusion}
\label{sec:conclusion}

We have given a fully analytic Bell-functional separation between qubit POVMs and all qubit-projective strategies with arbitrary shared two-qubit states. The proof reduces the global optimization to a transparent geometric mechanism: CHSH anchors the state and Bob's observables, while a weak ternary probe distinguishes a balanced three-ray qubit POVM from the one-diameter geometry of a ternary qubit PVM. The same argument proves an interval of fixed-qubit separating witnesses. For the rational instance \(\varepsilon=1/100\), an exact noncommutative sum-of-squares certificate further proves that the explicit qubit POVM reaches the unrestricted finite-dimensional quantum optimum. The result therefore separates the human-readable fixed-dimensional theorem from the exact global certificate and identifies the missing Naimark ancilla as an operational measurement resource under a dimension constraint.

\section*{Data Availability}

A Lean~4 formalization for Theorem \ref{thm:fixed-separation} is available in the \href{https://github.com/real-lin-zhu/Exact-Fixed-Dimension-Bell-Separation-between-Qubit-POVMs-and-PVMs/tree/main/Lean_formalisation}{\texttt{Lean\_formalisation}} of the \href{https://github.com/real-lin-zhu/Exact-Fixed-Dimension-Bell-Separation-between-Qubit-POVMs-and-PVMs}{public GitHub repository}.
The exact algebraic certificate, fixed word ordering, result files, and independent verification scripts supporting Theorem~\ref{thm:global-optimum} are available in the
\href{https://github.com/real-lin-zhu/Exact-Fixed-Dimension-Bell-Separation-between-Qubit-POVMs-and-PVMs}{same GitHub repository}. All SOS materials are contained in the
\href{https://github.com/real-lin-zhu/Exact-Fixed-Dimension-Bell-Separation-between-Qubit-POVMs-and-PVMs/tree/main/result/exact_sos}{\texttt{result/exact\_sos} directory}.

\section*{Acknowledgments}
This work was supported by the Guangdong Provincial Quantum Science Strategic Initiative under Grants No.~GDZX2503001 and No.~GDZX2403001.

\section*{Statement on the Use of Artificial Intelligence}

The authors formulated the research question and used several frontier LLMs under human direction for exploratory generation of candidate POVM constructions. The authors independently selected and reformulated the final Bell functional, derived the analytic fixed-qubit bound, checked every displayed calculation, constructed the exact sum-of-squares verification workflow, and reviewed the final manuscript. The machine-readable certificate was verified by exact arithmetic with zero coefficient residuals and exact positivity checks. The authors take full responsibility for the content and correctness of this article.

\bibliography{fixed_dimension_povm_pvm_pra}

\clearpage
\onecolumngrid
\appendix

\section{Exact verification and explicit form of the level-3 SOS certificate}
\label{app:exact-sos-proof}

This appendix proves Proposition~\ref{prop:exact-sos} and records the exact
algebraic results entering the certificate. For a noncommutative polynomial
\(p\), \(\operatorname{red}(p)\) denotes the deterministic reduced representative
used by the verifier. The map is extended linearly from words, moves every
Bob generator \(Q_y\) to the right using Alice--Bob commutation, removes
adjacent repetitions using idempotence, sets words containing
\(R_0R_1\) or \(R_1R_0\) to zero, and repeats these sound quotient-algebra
reductions until no rule applies. It does not interchange distinct measurement
settings belonging to the same party. The proof uses this fixed reduction
convention only to verify that the difference between the two sides of the
claimed identity reduces exactly to zero; no claim of a unique canonical normal
form is required. All scalar calculations are performed in the ordered field
\[
 \Ksqrt=\{a+b\sqrt2:a,b\in\mathbb Q\},
\]
with the physical real embedding \(\sqrt2>0\). The certificate fixes the
ordering of the reduced words and gives the entries of all matrices below
exactly in this field.

\subsection{Reduced word space}

Consider the unital \(*\)-algebra generated by the self-adjoint symbols
\(P_0,P_1,R_0,R_1,Q_0,Q_1\), subject to
\begin{align}
 P_x^2&=P_x, &
 Q_y^2&=Q_y, &
 R_a^2&=R_a, \notag\\
 R_0R_1&=R_1R_0=0, &
 [P_x,Q_y]&=[R_a,Q_y]=0,
 \label{eq:app-algebra-relations}
\end{align}
where \(x,y,a\in\{0,1\}\).  No commutation relation is imposed between
different settings of the same party.

Under the fixed reduction convention, Alice--Bob commutation writes every nonzero word in the form
\(w=w_Aw_B\), where \(w_A\) contains only \(P_0,P_1,R_0,R_1\) and
\(w_B\) contains only \(Q_0,Q_1\).  Adjacent repetitions are removed by
idempotence, words containing \(R_0R_1\) or \(R_1R_0\) vanish, and every
nonempty Bob word alternates between \(Q_0\) and \(Q_1\).
Let \(\operatorname{red}\) denote the linear map that sends a word
polynomial to this reduced representative.

Let \(a_n\) and \(b_n\) be the numbers of reduced Alice and Bob words of
length \(n\), respectively.  Up to degree six, these numbers are
\begin{equation}
 (a_0,\ldots,a_6)=(1,4,10,26,66,170,434),
 \qquad
 (b_0,\ldots,b_6)=(1,2,2,2,2,2,2).
 \label{eq:app-local-word-counts}
\end{equation}
Thus the numbers
\(N_n=\sum_{k=0}^n a_kb_{n-k}\) of distinct reduced words produced by this
convention at total degree \(n\) are
\begin{equation}
 (N_0,\ldots,N_6)=(1,6,20,56,148,384,988).
 \label{eq:app-degree-counts}
\end{equation}
In particular, the certificate uses
\begin{equation}
 N_0+N_1+N_2+N_3=83
 \label{eq:app-word-count}
\end{equation}
distinct reduced words of degree at most three. Fix the ordering supplied
with the certificate and write
\begin{equation}
 \boldsymbol v=(w_1,\ldots,w_{83})^{\mathsf T},
 \qquad w_1=\id .
 \label{eq:app-word-vector}
\end{equation}

\subsection{Exact polynomial identity}

Let
\[
 T:=\Omega\id-\Bellop,
 \qquad
 \Omega=2\sqrt2+\frac1{100}.
\]
Using the projector form of \(\Bellop\) in
Eq.~\eqref{eq:Belloperator-projectors}, the target polynomial is exactly
\begin{align}
 T={}&
 \left(2\sqrt2-\frac{501}{250}\right)\id
 +4P_0+\frac{1003}{250}Q_0+\frac2{125}Q_1
 \notag\\
 &-4P_0Q_0-4P_0Q_1-4P_1Q_0+4P_1Q_1
 \notag\\
 &+\frac3{125}R_0+\frac2{125}R_1
 -\frac4{125}R_0Q_0-\frac2{125}R_0Q_1
 -\frac4{125}R_1Q_1 .
 \label{eq:app-target-polynomial}
\end{align}
Thus \(T\) has exactly \(13\) nonzero reduced-word coefficients.

Let \(\mathcal R_6\) be the set of all oriented reduced words of degree at
most six.  Equation~\eqref{eq:app-degree-counts} gives
\begin{equation}
 \abs{\mathcal R_6}
 =\sum_{n=0}^6N_n=1603.
 \label{eq:app-degree-six-count}
\end{equation}
Among these words, \(163\) are fixed by the involution
\(u\mapsto u^\dagger\).  The number of adjoint orbits is therefore
\begin{equation}
 \frac{1603+163}{2}=883.
 \label{eq:app-adjoint-orbit-count}
\end{equation}

For \(u\in\mathcal R_6\), define
\begin{align}
 d_{ij}(u)
 &:=[u]\,\operatorname{red}(w_i^\dagger w_j)\in\mathbb Z,
 \notag\\
 t(u)
 &:=[u]\,\operatorname{red}(T)\in\Ksqrt ,
 \label{eq:app-coefficient-definitions}
\end{align}
where \([u]\) denotes coefficient extraction.  For a real symmetric
\(\Gram=(\Gram_{ij})\in\Ksqrt^{83\times83}\),
\begin{equation}
 \operatorname{red}
 \bigl(\boldsymbol v^\dagger\Gram\boldsymbol v\bigr)
 =
 \sum_{u\in\mathcal R_6}
 \left(
  \sum_{i,j=1}^{83}\Gram_{ij}d_{ij}(u)
 \right)u .
 \label{eq:app-Gram-reduction}
\end{equation}
The exact certificate satisfies
\begin{equation}
 t(u)=
 \sum_{i,j=1}^{83}\Gram_{ij}d_{ij}(u),
 \qquad u\in\mathcal R_6.
 \label{eq:app-coefficient-matching}
\end{equation}
These are \(1603\) oriented-word equalities.  Since both sides are
self-adjoint with real coefficients, it is equivalent to check one
representative of each of the \(883\) adjoint orbits.  Each equality is
checked separately for the rational and \(\sqrt2\) components, and every
residual is exactly zero.  The resulting Gram matrix has \(2321\) nonzero
upper-triangular entries.  Equation
\eqref{eq:app-coefficient-matching} proves the exact identity
\begin{equation}
 \Omega\id-\Bellop
 =\boldsymbol v^\dagger\Gram\boldsymbol v
 \label{eq:app-polynomial-identity}
\end{equation}

\subsection{Exact positivity and rank}

The exact matrices
\begin{equation}
 W\in\Ksqrt^{83\times61},
 \qquad
 H\in\Ksqrt^{61\times61}
 \label{eq:app-WH-dimensions}
\end{equation}
satisfy
\begin{equation}
 \Gram=WHW^{\mathsf T},
 \label{eq:app-facial-factorization}
\end{equation}
where \(W\) has full column rank and \(H\) is real symmetric.  There is also
an invertible rational upper-triangular matrix
\(U\in\mathbb Q^{61\times61}\).  Define
\begin{equation}
 \widetilde H:=U^{\mathsf T}HU
 =(\widetilde h_{ij})_{i,j=1}^{61}.
 \label{eq:app-congruence}
\end{equation}
For each \(i=1,\ldots,61\), exact arithmetic in the physical ordering of
\(\Ksqrt\) gives
\begin{equation}
 \delta_i:=
 \widetilde h_{ii}
 -\sum_{j\ne i}\abs{\widetilde h_{ij}}>0.
 \label{eq:app-diagonal-dominance}
\end{equation}
The sign of \(a+b\sqrt2\) is decided without approximation: when \(a\) and
\(b\) have opposite signs, it is enough to compare the rational numbers
\(a^2\) and \(2b^2\).

For every \(i<j\), set
\(s_{ij}:=\operatorname{sgn}(\widetilde h_{ij})\).  In the present
certificate all \(1830=\binom{61}{2}\) off-diagonal entries are nonzero.
Let \(e_i\) be the \(i\)th standard basis vector of \(\mathbb R^{61}\).
Entrywise comparison gives the exact decomposition
\begin{align}
 \widetilde H={}&
 \sum_{i=1}^{61}\delta_i e_i e_i^{\mathsf T}
 \notag\\
 &+
 \sum_{1\le i<j\le61}
 \abs{\widetilde h_{ij}}
 (e_i+s_{ij}e_j)(e_i+s_{ij}e_j)^{\mathsf T}.
 \label{eq:app-dd-decomposition}
\end{align}
All \(61+1830=1891\) weights are strictly positive in the physical ordering.
It follows exactly that
\begin{equation}
 \widetilde H\succ0,\qquad
 H=(U^{-1})^{\mathsf T}\widetilde H U^{-1}\succ0,\qquad
 \Gram\succeq0.
 \label{eq:app-exact-positivity}
\end{equation}
Since \(W\) has full column rank,
\begin{equation}
 \operatorname{rank}\Gram=61,
 \qquad
 \dim\ker\Gram=22.
 \label{eq:app-exact-rank}
\end{equation}

\subsection{Exact operator sum of squares}

The matrices above determine the Hermitian-square polynomials explicitly.
Define
\begin{equation}
 \boldsymbol f
 :=U^{-1}W^{\mathsf T}\boldsymbol v
 =(f_1,\ldots,f_{61})^{\mathsf T}.
 \label{eq:app-sos-polynomials}
\end{equation}
Each \(f_i\) is a polynomial of degree at most three with coefficients in
\(\Ksqrt\).  Because
\(H=(U^{-1})^{\mathsf T}\widetilde H U^{-1}\),
\begin{align}
 \boldsymbol v^\dagger\Gram\boldsymbol v
 &=
 \boldsymbol v^\dagger
 W(U^{-1})^{\mathsf T}\widetilde H
 U^{-1}W^{\mathsf T}\boldsymbol v
 \notag\\
 &=\boldsymbol f^\dagger\widetilde H\boldsymbol f.
 \label{eq:app-congruence-on-polynomials}
\end{align}
Combining Eqs.~\eqref{eq:app-polynomial-identity},
\eqref{eq:app-dd-decomposition}, and
\eqref{eq:app-congruence-on-polynomials} yields the exact weighted
operator-SOS identity
\begin{equation}
 \begin{aligned}
  \Omega\id-\Bellop
  ={}&
  \sum_{i=1}^{61}\delta_i f_i^\dagger f_i
  \\
  &+
  \sum_{1\le i<j\le61}
  \abs{\widetilde h_{ij}}\,
  (f_i+s_{ij}f_j)^\dagger
  (f_i+s_{ij}f_j).
 \end{aligned}
 \label{eq:app-explicit-weighted-sos}
\end{equation}
This is the exact SOS form over the ordered field \(\Ksqrt\): it is a sum of
\(1891\) Hermitian squares with exact positive weights in \(\Ksqrt\).

For the usual unweighted notation, define
\begin{equation}
 L_i:=\sqrt{\delta_i}\,f_i,
 \qquad
 L_{ij}:=
 \sqrt{\abs{\widetilde h_{ij}}}\,
 (f_i+s_{ij}f_j).
 \label{eq:app-unweighted-polynomials}
\end{equation}
After adjoining the positive square roots of the weights,
Eq.~\eqref{eq:app-explicit-weighted-sos} becomes
\begin{equation}
 \Omega\id-\Bellop
 =
 \sum_{i=1}^{61}L_i^\dagger L_i
 +
 \sum_{1\le i<j\le61}L_{ij}^\dagger L_{ij}.
 \label{eq:app-standard-sos}
\end{equation}
The weighted form in Eq.~\eqref{eq:app-explicit-weighted-sos} is the form
whose coefficients remain entirely in \(\Ksqrt\).

Finally, Eq.~\eqref{eq:app-explicit-weighted-sos} proves
\(\Bellop\le\Omega\id\) in every representation of the projector algebra.
Together with the explicit saturating strategy and local Naimark dilation,
the exact dimension-unrestricted value is therefore
\begin{equation}
 \max_{p\in\Q}\mathcal B(p)
 =2\sqrt2+\frac1{100}.
 \label{eq:app-exact-quantum-value}
\end{equation}
This completes the proof of Proposition~\ref{prop:exact-sos}. 

All exact SOS certificates, associated result files, the fixed word ordering, and verification scripts
are available in the
\href{https://github.com/real-lin-zhu/Exact-Fixed-Dimension-Bell-Separation-between-Qubit-POVMs-and-PVMs/tree/main/result/exact_sos}
{\texttt{exact\_sos} directory}
of the public GitHub repository.

\end{document}